\pdfoutput=1
\ifx\synctex\undefined\else\synctex=1\fi
\documentclass[a4paper,envcountsame,oribibl]{llncs}

\usepackage[utf8]{inputenc}

\usepackage{amsmath}
\usepackage{amssymb}
\usepackage{stmaryrd} %
\usepackage{tikz}
\usepackage{forest}
\usepackage{subcaption}

\usetikzlibrary{automata,positioning}

\usepackage{paralist}

\usepackage[pdftex,%
  colorlinks,%
  hidelinks,%
  hypertexnames=false,%
  pdfauthor={Pierre Ganty, Elena Gutierrez},%
  pdftitle={Parikh image of Pushdown Automata}]{hyperref}

\def\len#1{{\vert{#1}\vert}}
\def\seq#1{{\overline{#1}}}
\newcommand{\state}{\mathit{state}}
\newcommand{\tape}{\mathit{tape}}
\newcommand{\stack}{\mathit{stack}}

\title{Parikh Image of Pushdown Automata}
\author{Pierre Ganty\inst{1}\fnmsep\thanks{Pierre Ganty has been supported by the Madrid Regional Government project S2013/ICE-2731, \emph{N-Greens Software - Next-GeneRation Energy-EfficieNt Secure Software}, and the Spanish Ministry of Economy and Competitiveness project No. TIN2015-71819-P, \emph{RISCO - RIgorous analysis of Sophisticated COncurrent and distributed systems}.} \and Elena Gutiérrez\inst{1,2}\fnmsep\thanks{Elena Gutiérrez is partially supported by BES-2016-077136 grant from the Spanish Ministry of Economy, Industry and Competitiveness.}}
\institute{IMDEA Software Institute, Madrid, Spain \and Universidad Politécnica de Madrid, Spain  \email{\{pierre.ganty,elena.gutierrez\}@imdea.org}}

\begin{document}

\pagestyle{headings}  %

\maketitle
\begin{abstract}
We compare pushdown automata (PDAs for short) against other representations.
First, we show that there is a family of PDAs over a unary alphabet with \(n\) states and \(p \geq 2n + 4\) stack symbols that accepts one single long word for which every equivalent context-free grammar needs \(\Omega(n^2(p-2n-4))\) variables.
This family shows that the classical algorithm for converting a PDA to an equivalent context-free grammar is optimal even when the alphabet is unary.
Moreover, we observe that language equivalence and Parikh equivalence, which ignores the ordering between symbols, coincide for this family. 
We conclude that, when assuming this weaker equivalence, the conversion algorithm is also optimal.
Second, Parikh's theorem motivates the comparison of PDAs against finite state automata.
In particular, the same family of unary PDAs gives a lower bound on the number of states of every Parikh-equivalent finite state automaton. 
Finally, we look into the case of unary deterministic PDAs.
We show a new construction converting a unary deterministic PDA into an \hyphenation{e-qui-va-lent}{equivalent} context-free grammar that achieves best known bounds.

\end{abstract}

\section{Introduction}\label{sec:intro}

Given a context-free language which representation, pushdown automata or context-free grammars, is more concise?   
This was the main question studied by Goldstine et al.~\cite{GOLD} in a paper where they introduced an infinite family of context-free languages whose representation by a pushdown automaton is more concise than by context-free grammars. 
In particular, they showed that each language of the family is accepted by a pushdown automaton with \(n\) states and \(p\) stack symbols, but every context-free grammar needs at least \(n^2p + 1\) variables if \(n>1\) (\(p\) if \(n=1\)).
Incidentally, the family shows that the translation of a pushdown automaton into an equivalent context-free grammar used in textbooks~\cite{AUTTH}, which uses the same large number of \(n^2p + 1\) variables if \(n>1\) (\(p\) if \(n=1\)), is optimal in the sense that there is no other algorithm that always produces fewer grammar variables.

Today we revisit these questions but this time we turn our attention to the unary case.
We define an infinite family of context-free languages as Goldstine et al. did but our family differs drastically from theirs.   
Given \(n\geq 1\) and \(k\geq 1\), each member of our family is given by a PDA with \(n\) states, \(p = k + 2n + 4\) stack symbols and a \emph{single input symbol}.%
\footnote{Their family has an alphabet of non-constant size.}
We show that, for each PDA of the family, every equivalent context-free \hyphenation{gra-mmar}{grammar} has \(\Omega({n^2(p -2n - 4)})\) variables.
Therefore, this family shows that the textbook translation of a PDA into a language-equivalent context-free grammar is \emph{optimal}\footnote{In a sense that we will precise in Section \ref{sec:CFG} (Remark \ref{rem:CFG}).} even when the alphabet is unary. 
Note that if the alphabet is a singleton, equality over words (two words are equal if the same symbols appear at the same positions) coincides with Parikh equivalence (two words are Parikh-equivalent if each symbol occurs equally often in both words%
\footnote{But not necessarily at the same positions, e.g. \(ab \) and \(ba\) are Parikh-equivalent.}).
Thus, we conclude that the conversion algorithm is also optimal for Parikh equivalence. We also investigate the special case of deterministic PDAs over a singleton alphabet for which equivalent context-free grammar \hyphenation{re-pre-sen-ta-tions}{representations} of small size had been defined~\cite{Chistikov2014,Pighizzini2009}. 
We give a new definition for an equivalent context-free grammar given a unary deterministic PDA. 
Our definition is constructive (as far as we could tell the result of Pighizzini~\cite{Pighizzini2009} is not) and achieves the best known bounds~\cite{Chistikov2014} by combining two known constructions.

Parikh's theorem \cite{Parikh66} states that every context-free language has the same Parikh image as some regular language. 
This allows us to compare PDAs against finite state automata (FSAs for short) for Parikh-equivalent languages.
First, we use the same family of PDAs to derive a lower bound on the number of states of every Parikh-equivalent FSA.
The comparison becomes simple as its alphabet is unary and it accepts one single word.
Second, using this lower bound we show that the 2-step procedure chaining existing constructions: 
\begin{inparaenum}[\upshape(\itshape i\upshape)]
\item translate the PDA into a language-equivalent context-free grammar~\cite{AUTTH}; and
\item translate the context-free grammar into a Parikh-equivalent FSA~\cite{IPL}  
\end{inparaenum}
yields \emph{optimal}\footnote{In a sense that we will precise in Section \ref{sec:FSA} (Remark \ref{rem:FSA}).} results in the number of states of the resulting FSA.

As a side contribution, we introduce a semantics of PDA runs as trees that we call \emph{actrees}.
The richer tree structure (compared to a sequence) makes simpler to compare each PDA of the family with its smallest grammar representation.
\paragraph{Structure of the paper.} After preliminaries in Section~\ref{sec:prelim} we introduce the tree-based semantics in~\ref{sec:actrees}.
In Section~\ref{sec:CFG} we compare PDAs and context-free grammars when they represent Parikh-equivalent languages.
We will define the infinite family of PDAs and establish their main properties.
We dedicate Section~\ref{sec:specialcase} to the special case of deterministic PDAs over a unary alphabet.
Finally, Section~\ref{sec:FSA} focuses on the comparison of PDAs against finite state automata for Parikh-equivalent languages.

\section{Preliminaries}\label{sec:prelim}
A \emph{pushdown automaton} (or PDA) is a 6-tuple \( (Q,\Sigma,\Gamma,\delta,q_0,Z_0)\) where \(Q\) is a finite nonempty set of \emph{states} including \(q_0\), the \emph{initial} state; \(\Sigma\) is the \emph{input alphabet}; \(\Gamma\) is the \emph{stack alphabet} including \(Z_0\), the \emph{initial} stack symbol; and \(\delta\) is a finite subset of \(Q\times\Gamma\times(\Sigma\cup\{\varepsilon\})\times Q\times\Gamma^*\) called the \emph{actions}. 
We write \((q, X) \hookrightarrow_b (q',\beta)\) to denote an action \((q,X,b,q',\beta) \in \delta\). 
We sometimes omit the subscript to the arrow.

An \emph{instantaneous description} (or ID) of a PDA is a pair $(q,\beta) $ where $q\in Q$ and $\beta \in \Gamma^{*}$.
We call the first component of an ID the \emph{state} and the second the \emph{stack content}.
The \emph{initial ID} consists of the initial state and the initial stack symbol for the stack content.
When reasoning formally, we use the functions \(\state\) and \(\stack\) which, given an ID, returns its state and stack content, respectively.

An action \( (q,X) \hookrightarrow_b (q',\beta)\) is \emph{enabled} at ID \(I\) if \(\state(I)=q\) and \(  (\,\stack(I)\,)_1=X\).%
\footnote{\( (w)_i \) is the \(i\)-th symbol of \(w\) if \(1\leq i \leq \len{w}\); else \( (w)_i=\varepsilon\). \(\len{w}\) is the length of \(w\).}
Given an ID \((q,X\gamma)\) enabling \( (q,X) \hookrightarrow_b (q',\beta)\), define the \emph{successor ID} to be \( (q',\beta\gamma) \).
We denote this fact as $(q,X\gamma )\vdash_b (q',\beta\gamma)$,
and call it a \emph{move} that \emph{consumes} \(b\) from the input.%
\footnote{When \(b=\varepsilon\) the move does not consume input.}
We sometimes omit the subscript of \(\vdash\) when the input consumed (if any) is not important.
Given \(n\geq0\), a \emph{move sequence}, denoted \( I_0 \vdash_{b_1} {\cdots}\vdash_{b_n} I_n \), is a finite sequence of IDs  \(I_0 I_1\ldots I_n\) such that \(I_i\vdash_{b_i} I_{i+1}\) for all \(i\). 
The move sequence \emph{consumes} \(w\) (\emph{from the input}) when \(b_1 \cdots b_n = w\). 
We concisely denote this fact as \(I_0 \vdash {\stackrel{w}{\ldots}} \vdash I_n\). A move sequence $I \vdash{\cdots}\vdash I'$ is a \emph{quasi-run} when \(\len{\stack(I)}=1\) and \(\len{\stack(I')}=0\); and a \emph{run} when, furthermore, \(I\) is the initial ID. 
Define the \emph{language} of a PDA \(P\) as 
\(L(P) = \{w \in\Sigma^{*} \mid P \text{ has a run consuming } w \}\).

The \emph{Parikh image} of a word \(w\) over an alphabet \(\{b_1, \ldots, b_n\}\), denoted by \(\lbag w \rbag\), is the vector \((x_1,\ldots, x_n) \in \mathbb{N}^n\) such that \(x_i\) is the number of occurrences of \(b_i\) in \(w\). 
The \emph{Parikh image of a language} \(L\), denoted by \(\lbag L\rbag\), is the set of Parikh images of its words. 
When \(\lbag L_1 \rbag = \lbag L_2 \rbag\), we say \(L_1\) and \(L_2\) are \emph{Parikh-equivalent}. 

We assume the reader is familiar with the basics of \emph{finite state automata} (or FSA for short) and \emph{context-free grammars} (or CFG). 
Nevertheless we fix their notation as follows.
We denote a FSA as a tuple \((Q,\Sigma,\delta,q_0,F)\) where \(Q\) is a finite set of \emph{states} including the \emph{initial} state \(q_0\) and the \emph{final} states \(F\); \(\Sigma\) is the \emph{input alphabet} and \(\delta \subseteq Q \times (\Sigma \cup \{\varepsilon\}) \times Q\) is the \emph{set of transitions}.
We denote a CFG as a tuple \( (V,\Sigma, S, R)\) where \(V\) is a finite set of \emph{variables} including \(S\) the \emph{start} variable, \(\Sigma\) is the \emph{alphabet} or set of terminals and \(R \subseteq V \times (V \cup \Sigma)^* \) is a finite set of \emph{rules}. 
Rules are conveniently denoted \(X \rightarrow \alpha\). 
Given a FSA \(A\) and a CFG \(G\) we denote their \emph{languages} as \(L(A)\) and \(L(G)\), respectively.

Finally, let us recall the translation of a PDA into an equivalent CFG.

Given a PDA \(P=(Q,\Sigma,\Gamma,\delta,q_0,Z_0)\), define the CFG \(G = (V,\Sigma,R,S)\) where 
\begin{compactitem}
	\item The set \(V\) of variables  --- often called the \emph{triples} --- is given by 
		\begin{equation}
			\{ [qXq'] \mid q,q'\in Q, X\in\Gamma \} \cup\{S\}\enspace .\label{eq:triples}
		\end{equation}
	\item The set \(R\) of production rules is given by
\begin{equation}
	\begin{split}
	&\{ S \rightarrow [q_0 Z_0 q] \mid q\in Q \} \\
	{} \cup {} &\{ [qXr_{d}] \rightarrow b [q'(\beta)_1 r_1]\ldots[r_{d-1}(\beta)_{d} r_{d}] \\
		&\qquad \mid (q,X) \hookrightarrow_b (q',\beta),d=\len{\beta}, r_1,\ldots,r_{d}\in Q \}
	\end{split}\label{eq:rulesoftriples}
\end{equation}
\end{compactitem}

For a proof of correctness, see the textbook of Ullman et al.~\cite{AUTTH}.
The previous definition easily translates into a \emph{conversion algorithm}.
Observe that the runtime of such algorithm depends polynomially on \(\len{Q}\) and \(\len{\Gamma}\), but exponentially on \(\len{\beta}\).
\section{A Tree-Based Semantics for Pushdown Automata}\label{sec:actrees}

In this section we introduce a tree-based semantics for PDA.
Using trees instead of sequences sheds the light on key properties needed to present our main results.

Given an action \(a\) denoted by \((q, X) \hookrightarrow_b (q', \beta)\), \(q\) is the \emph{source} state of \(a\), \(q'\) the \emph{target} state of \(a\), \(X\) the symbol \(a\) \emph{pops} and \( \beta \) the (possibly empty) sequence of symbols \(a\) \emph{pushes}.  

A \emph{labeled tree} \(c(t_1, \ldots , t_k)\) \((k \geq 0)\) is a finite tree whose nodes are labeled, where \(c\) is the label of the root and \(t_1,\ldots, t_k\) are labeled trees, the children of the root. 
When \(k=0\) we prefer to write \(c\) instead of \(c()\).
Each labeled tree \(t\) defines a sequence, denoted \(\seq{t}\), obtained by removing the symbols ‘(’, ‘)’ or ‘,’ when interpreting \(t\) as a string, e.g. \(\seq{c(c_1,c_2(c_{21}))} = c\, c_1\, c_2\, c_{21}\).
The \emph{size} of a labeled tree \(t\), denoted \(\len{t}\), is given by \(\len{\seq{t}}\).
It coincides with the number of nodes in \(t\).

\begin{definition}
\label{def:tree}
Given a PDA \(P\), an action-tree (or \emph{actree} for short) is a labeled tree \(a(a_1(\ldots),\ldots,a_d(\ldots))\) where \(a\) is an action of \(P\) pushing \(\beta\) with \(\len{\beta}=d\) and
each children \(a_i(\ldots)\) is an actree such that \(a_i\) pops \( (\beta)_i \) for all \(i\). %
Furthermore, an actree \(t\) must satisfy that the source state of \( (\seq{t})_{i+1}\) and the target state of \( (\seq{t})_{i}\) coincide for every \(i\).

An actree \(t\) consumes an input resulting from replacing each action in the sequence \(\seq{t}\) by the symbol it consumes (or \(\varepsilon\), if the action does not consume any).
An actree \(a(\ldots)\) is \emph{accepting} if the initial ID enables \(a\).
\end{definition}
\begin{example}
\label{ex:tree}
Consider a PDA \(P\) with actions \(a_1\) to \(a_5\) respectively given by \((q_0,  X_1)\hookrightarrow_\varepsilon (q_0, X_0\,X_0)\), \((q_0, X_0)\hookrightarrow_\varepsilon (q_1, X_1\,\star)\), \((q_1,  X_1)\hookrightarrow_\varepsilon (q_1, X_0\,X_0)\),\\\mbox{ \((q_1, X_0)\hookrightarrow_b (q_1, \varepsilon)\)} and \((q_1, \star)\hookrightarrow_\varepsilon (q_0, \varepsilon)\). 
The reader can check that the actree \( t = a_1( a_2( a_3( a_4 , a_4 ) , a_5) , a_2( a_3( a_4 , a_4 ) , a_5) )\), depicted in Figure~\ref{fig:tree}, satisfies the conditions of Definition~\ref{def:tree} where \(\seq{t}=a_1\,a_2\,a_3\,a_4\,a_4\,a_5\,a_2\,a_3\,a_4\,a_4\,a_5\), \(\len{t}=11\) and the input consumed is \(b^4\).
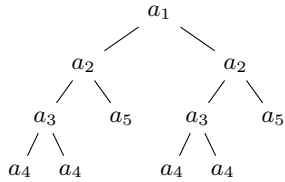
\begin{figure}[!ht]
\centering
\begin{tikzpicture}[level/.style={sibling distance = 2cm/#1,
  level distance = .7cm}]
	\node { \(a_1\) } 
	child { node { \(a_2\) } 
					child { node { \(a_3\) } 
									child { node { \(a_4\) } }
									child { node { \(a_4\) } }
					}
					child { node { \(a_5\) } }
	} 
	child { node { \(a_2\) } 
					child { node { \(a_3\) } 
									child { node { \(a_4\) } }
									child { node { \(a_4\) } }
					}
					child { node { \(a_5\) } }
	};
\end{tikzpicture}
\caption{Depiction of the tree \( a_1( a_2( a_3( a_4 , a_4 ) , a_5) , a_2( a_3( a_4 , a_4 ) , a_5) )\)}
\label{fig:tree}
\end{figure}
\end{example}

We recall the notion of \emph{dimension} of a labeled tree \cite{Esparza2014} and we relate dimension and size of labeled trees in Lemma~\ref{lemma:length}.
\begin{definition}
\label{def:dim}
The \emph{dimension} of a labeled tree \(t\), denoted as \(d(t)\), is inductively defined as follows. \(d(t)=0\) if \(t=c\), otherwise we have \(t=c(t_1,\ldots,t_k)\) for some \(k>0\) and
  \begin{equation*}
    d(t) = 
      \begin{cases}
        max_{i \in\{1,...,k\}} d(t_i) & \text{if there is a unique maximum},\\
        max_{i \in\{1,...,k\}} d(t_i)+1 & \text{otherwise}.
      \end{cases}
  \end{equation*}
\end{definition}

\begin{example} 
	The annotation \(\stackrel{d(t)}{t}\!\!(\ldots)\) shows
	the actree of Example~\ref{ex:tree} has dimension \(2\)
	\[ \stackrel{2}{a_1}( \stackrel{1}{a_2}( \stackrel{1}{a_3}( \stackrel{0}{a_4} , \stackrel{0}{a_4} ) , \stackrel{0}{a_5}) , \stackrel{1}{a_2}( \stackrel{1}{a_3}( \stackrel{0}{a_4} , \stackrel{0}{a_4} ) , \stackrel{0}{a_5}) )\enspace .\] \end{example}

\begin{lemma}
\(\len{t} \geq 2^{d(t)}\) for every labeled tree \(t\).
\label{lemma:length}
\end{lemma}

The proof of the lemma is given in the Appendix.
The actrees and the quasi-runs of a PDAs are in one-to-one correspondence as reflected in Theorem~\ref{thm:corespondence} whose proof is in the Appendix.
\begin{theorem}
\label{thm:corespondence}
Given a PDA, its actrees and quasi-runs are in a one-to-one correspondence.
 \end{theorem}

\section{Parikh-Equivalent Context-free Grammars}\label{sec:CFG}
In this section we compare PDAs against CFGs when they describe Parikh-equivalent languages.
We first study the general class of (nondeterministic) PDAs and, in Section~\ref{sec:specialcase}, we look into the special case of unary deterministic PDAs.

We prove that, for every \(n \geq 1\) and \(p \geq 2n + 4\), there exists a PDA with \(n\) states and \(p\) stack symbols for which every Parikh-equivalent CFG has \(\Omega({n^2(p -2n - 4)})\) variables.
To this aim, we present a family of PDAs \(P(n,k)\) where \(n\geq1\) and \(k\geq1\).
Each member of the family has \(n\) states and \(k + 2n + 4\) stack symbols, and accepts one single word over a unary input alphabet.

\subsection{The Family \texorpdfstring{\(P(n,k)\)}{P(n,k)} of PDAs}
\begin{definition}
\label{def:P}
Given natural values $n\geq1$ and $k\geq 1$, define the PDA $P(n,k)$ with states $Q=\{{q_{i}\mid {0\leq i \leq n-1}}\}$, input alphabet $\Sigma=\{b\}$, stack alphabet $\Gamma=\{S,\,\star,\,\$ \}\cup\{X_i \mid 0\leq i\leq k\}\cup \{s_i \mid 0 \leq i\leq n-1 \} \cup \{r_i \mid 0 \leq i\leq n-1 \}$, initial state \(q_0\), initial stack symbol \(S\) and actions \(\delta\)
\[
\begin{array}[t]{r@{\;}c@{\;}lr}
(q_0,S) &\hookrightarrow_b& (q_0, X_{k}\,r_0) & \\
(q_i,X_{j}) &\hookrightarrow_b& (q_i, X_{j-1}\,r_m\,s_i\,X_{j-1}\,r_m) & \forall\,i,m \in \{0, \ldots, n-1\}, \forall \, j \in \{1, \ldots,k\},\\
(q_j, s_i) &\hookrightarrow_b& (q_i, \varepsilon) & \forall i,j \in \{0, \ldots, n-1\},\\
(q_i, r_i) &\hookrightarrow_b& (q_i, \varepsilon) & \forall i \in \{0, \ldots, n-1\},\\
(q_i, X_0) &\hookrightarrow_b& (q_i, X_{k}\,\star) & \forall i \in \{0, \ldots, n-1\},\\
(q_i, X_0) &\hookrightarrow_b& (q_{i+1}, X_{k}\,\$) & \forall i \in \{0, \ldots, n-2\},\\
(q_i, \star) &\hookrightarrow_b& (q_{i-1}, \varepsilon) & \forall i \in \{1, \ldots, n-1\},\\
(q_0, \$) &\hookrightarrow_b& (q_{n-1}, \varepsilon) & \\
(q_{n-1}, X_0) &\hookrightarrow_b& (q_{n-1}, \varepsilon)\enspace
\end{array}
\]
\end{definition}

\begin{lemma}
\label{lemma:unique_w}
Given \(n\geq 1\) and \(k\geq1\), \(P(n,k)\) has a single accepting actree consuming input \(b^{N}\) where \(N\geq 2^{n^2\,k}\).
\end{lemma}

\begin{proof}

Fix values \(n\) and \(k\) and refer to the member of the family \(P(n,k)\) as \(P\).
We show that \(P\) has exactly one accepting actree.
We define a witness labeled tree \(t\) inductively on the structure of the tree.
Later we will prove that the induction is finite.
First, we show how to construct the root and its children subtrees.
This corresponds to case 1 below.
Then, each non-leaf subtree is defined inductively in cases 2 to 5.
Note that each non-leaf subtree of \(t\) falls into one (and only one) of the cases.
In fact, all cases are disjoint, in particular 2, 4 and 5.
The reverse is also true: all cases describe a non-leaf subtree that does occur in \(t\).
Finally, we show that each case describes \emph{uniquely} how to build the next layer of children subtrees of a given non-leaf subtree.

\begin{compactenum}
  \item \(t=a(a_1(\ldots), a_2)\) where \(a = (q_0,S)\hookrightarrow_b (q_0, X_{k}\,r_0)\) and \(a_1(\ldots)\) and \(a_2\) are of the form:
  \begin{align*}
      a_2 &= (q_0, r_0)\hookrightarrow_b (q_0, \varepsilon) &\text{only action popping \(r_0\)}\\
      a_1 &= (q_0, X_{k})\hookrightarrow_b (q_0, X_{k-1}\,r_0\,s_0\,X_{k-1}\,r_0)&\text{only way to enable  \(a_2\).}
    \end{align*}

   Note that the initial ID \((q_0, S)\) enables \(a\) which is the only action of \(P\) with this property. Note also that \(\stackrel{d}{a}(\stackrel{d}{a_1}(\ldots),\stackrel{0}{a_2})\) holds, where \(d> 0\).

	\item Each subtree whose root is labeled \(a = (q_i,X_j) \hookrightarrow_b (q_i, X_{j-1}\,r_m\,s_i\,X_{j-1}\,r_m)\) with \(i,m \in \{0,\ldots,n-1\}\) and \( j \in \{2, \ldots, k\}\) has the form \(a(a_1(\ldots),a_2,a_3,a_1(\ldots),a_2)\) where
		\begin{align*}
			a_2 &= (q_m, r_m)\hookrightarrow_b (q_m, \varepsilon) &\text{only action popping \(r_m\)}\\
			a_3 &= (q_m, s_i)\hookrightarrow_b (q_i,\varepsilon) &\text{only action popping \(s_i\) from \(q_m\)}\\
			a_1 &= (q_i, X_{j-1})\hookrightarrow_b (q_i, X_{j-2}\,r_m\,s_i\,X_{j-2}\,r_m)&\text{only way to enable  \(a_2\).}
		\end{align*}

		Assume for now that \(t\) is unique. Therefore, as the 1st and 4th child of \(a\) share the same label \(a_1\), they also root the same subtree. Thus, it holds (\(d > 0\)) \[\stackrel{d+1}{a}(\stackrel{d}{a_1}(\ldots),\stackrel{0}{a_2},\stackrel{0}{a_3},\stackrel{d}{a_1}(\ldots),\stackrel{0}{a_2})\enspace .\]
	
  \item Each subtree whose root is labeled \(a = (q_i, X_0)\hookrightarrow_b (q_{i+1}, X_{k}\,\$)\) with \(i \in \{0, \dots, n-2\}\) has the form \(a(a_1(\ldots),a_2)\) where
		\begin{align*}
			a_2 &= (q_0,\$)\hookrightarrow_b (q_{n-1},\varepsilon) &\text{only action popping \(\$\)}\\
			a_1 &= (q_{i+1},X_{k})\hookrightarrow_b (q_{i+1}, X_{k-1}\,r_0\,s_{i+1}\,X_{k-1}\,r_0) &\text{only way to enable \(a_2\).} 
		\end{align*}

    Note that \(\stackrel{d}{a}(\stackrel{d}{a_1}(\ldots),\stackrel{0}{a_2})\) holds, where \(d> 0\).
	\item Each subtree whose root is labeled \(a = (q_i,X_1) \hookrightarrow_b (q_i, X_0\,r_m\,s_i\,X_0\,r_m)\) with \(i \in \{0, \dots, n-1\}\) and \(m \in \{0,\ldots, n-2\}\) has the form \[a(a_1(a_{11}(\ldots),a_{12}),a_2,a_3,a_1(a_{11}(\ldots),a_{12}),a_2)\enspace .\] where
		\begin{align*}
			a_2 &= (q_m, r_m)\hookrightarrow_b (q_m, \varepsilon) &\text{only action popping \(r_m\)}\\
			a_3 &= (q_m, s_i)\hookrightarrow_b (q_i,\varepsilon) &\text{only action popping \(s_i\) from \(q_m\)}\\
			a_1 &= (q_i, X_0)\hookrightarrow_b (q_i, X_{k}\,\star) &\text{assume it for now}\\
			a_{12} &= (q_{m+1},\star)\hookrightarrow_b(q_m,\varepsilon) &\text{only way to enable \(a_2\)}\\
			a_{11} &= (q_i,X_{k})\hookrightarrow_b (q_i, X_{k-1}\, r_{m+1}\, s_i\, X_{k-1}\, r_{m+1}) &\text{only way to enable \(a_{12}\).}
		\end{align*}

		Assume \(a_1\) is given by the action \((q_i, X_0)\hookrightarrow_b (q_{i+1}, X_{k}\,\$)\) instead.
		Then following the action popping \(\$\), we would end up in the state \(q_{n-1}\), not enabling \(a_2\)  since \(m < n-1\).

    Again, assume for now that \(t\) is unique. Hence, as the 1st and 4th child of \(a\) are both labeled by \(a_1\), they root the same subtree. Thus, it holds (\(d>0\)) \[\stackrel{d+1}{a}(\stackrel{d}{a_1}(\stackrel{d}{a_{11}}(\ldots),\stackrel{0}{a_{12}}),\stackrel{0}{a_2},\stackrel{0}{a_3},\stackrel{d}{a_1}(\stackrel{d}{a_{11}}(\ldots),\stackrel{0}{a_{12}}),\stackrel{0}{a_2})\enspace .\]

	\item Each subtree whose root is labeled \(a = (q_i,X_1) \hookrightarrow_b (q_i, X_0\,r_{n-1}\,s_i\,X_0\,r_{n-1})\) with \(i \in \{0, \ldots, n-1\}\) has the form \(a(a_1(\ldots),a_2,a_3,a_1(\ldots),a_2)\) where
		\begin{align*}
			a_2 &= (q_{n-1}, r_{n-1})\hookrightarrow_b (q_{n-1},\varepsilon) &\text{only action popping \(r_{n-1}\)}\\
			a_3 &= (q_{n-1}, s_i)\hookrightarrow_b (q_i,\varepsilon)  &\text{only action popping \(s_i\) from \(q_i\)}\\
			a_1 &= \begin{cases}
          (q_i, X_0)\hookrightarrow_b (q_{i+1}, X_{k}\,\$)& \text{if}~i < n-1\\
          (q_{n-1}, X_0)\hookrightarrow_b (q_{n-1}, \varepsilon) & \text{otherwise}
          \end{cases} & \text{Assume it for now.}
		\end{align*}

		For both cases (\(i<n-1\) and \(i = n-1\)), assume \(a_1\) is given by \((q_i, X_0)\hookrightarrow_b(q_i, X_{k}\,\star)\) instead.
		Then, the action popping \(\star\) must end up in the state \(q_{n-1}\) in order to enable \(a_2\), i.e., it must be of the form \((q_n, \star) \hookrightarrow_b(q_{n-1},\varepsilon)\). 
		Hence the action popping \(X_{k}\) must be of the form \((q_i,X_{k}) \hookrightarrow_b (q_i, X_{k-1}\,r_m\,s_i\,X_{k-1}\,r_m)\) where necessarily \(m = n\), a contradiction (the stack symbol \(r_n\) is not defined in \(P\)).

Assume for now that \(t\) is unique. Then, as the 1st and 4th child of \(a\) are labeled by \(a_1\), they root the same subtree (possibly a leaf). Thus, it holds (\(d\geq 0\)) \[\stackrel{d+1}{a}(\stackrel{d}{a_1}(\ldots),\stackrel{0}{a_2},\stackrel{0}{a_3},\stackrel{d}{a_1}(\ldots),\stackrel{0}{a_2})\enspace .\]
\end{compactenum}

We now prove that \(t\) is finite by contradiction. 
Suppose \(t\) is an infinite tree. 
König's Lemma shows that \(t\) has thus at least one infinite path, say \(p\), from the root. 
As the set of labels of \(t\) is finite then some label must repeat infinitely often along \(p\). 
Let us define a strict partial order between the labels of the non-leaf subtrees of \(t\). 
We restrict to the non-leaf subtrees because no infinite path contains a leaf subtree. Let \(a_1(\ldots)\) and \(a_2(\ldots)\) be two non-leaf subtrees of \(t\). 
Let \(q_{i_1}\) be the source state of \(a_1\) and \(q_{f_1}\) be the target state of the last action in the sequence \(\seq{a_1(\ldots)}\).
Define \(q_{i_2}, q_{f_2}\) similarly for \(a_2(\ldots)\).
Let \(X_{j_1}\) be the symbol that \(a_1\) pops and \(X_{j_2}\) be the symbol that \(a_2\) pops.  Define \(a_1 \prec a_2\) if{}f
\begin{inparaenum}[\upshape(\itshape a\upshape)]
  \item either \(i_1 < i_2\),
  \item or \(i_1 = i_2\) and \(f_1 < f_2\),
  \item or \(i_1 = i_2, f_1 = f_2 \) and \(j_1 > j_2\).
\end{inparaenum}
First, note that the label \(a\) of the root of \(t\) (case 1) only occurs in the root as there is no action of \(P\) pushing \(S\).
Second, relying on cases 2 to 5, we observe that every pair of non-leaf subtrees \(a_1(\ldots)\) and \(a_2(\ldots)\) (excluding the root) such that \(a_1(\ldots)\) is the parent node of \(a_2(\ldots)\) verifies \mbox{\(a_1(\ldots) \prec a_2(\ldots)\)}.
Using the transitive property of the strict partial order \(\prec\), we conclude that \hyphenation{e-ve-ry}{every}pair of subtrees \( a_1(\ldots)\) and \( a_2(\ldots) \) in \(p\)  such that \( a_1(\ldots a_2(\ldots) \ldots) \) verifies \(a_1(\ldots) \prec a_2(\ldots) \). 
Therefore, no repeated variable can occur in \(p\) (contradiction).
We conclude that \(t\) is finite. 

The reader can observe that \(t = a(\ldots)\) verifies all conditions of the definition of actree (Definition~\ref{def:tree}) and the initial ID enables \(a\), thus it is an accepting actree of \(P\).
Since we also showed that no other tree can be defined using the actions of \(P\), \(t\) is unique.

Finally, we give a lower bound on the length of the word consumed by \(t\). 
To this aim, we prove that \(d(t) = n^2\,k\). 
Then since all actions consume input symbol \(b\), Lemma~\ref{lemma:length} shows that the word \(b^N\) consumed is such that \(N \geq 2^{n^2\,k}\).

Note that, if a subtree of \(t\) verifies case \(1\) or \(3\), its dimension remains the same w.r.t. its children subtrees.
Otherwise, the dimension always grows.
Recall that all cases from 1 to 5 describe a set of labels that does occur in \(t\).
Also, as \(t\) is unique, no path from the root to a leaf repeats a label.
Thus, to compute the dimension of \(t\) is enough to count the number of distinct labels of \(t\) that are included in cases \(2\), \(4\) and \(5\), which is equivalent to compute the size of the set
\[D = \{(q_i,X_j) \hookrightarrow (q_i, X_{j-1}\,r_m\,s_i\,X_{j-1}\,r_m) \mid 1 \leq j \leq k,\,0 \leq i, m \leq n-1\}\enspace .\]
Clearly \(\len{D} = n^2\,k\) from which we conclude that \(d(t) = n^2\,k\). Hence, \(\len{t} \geq 2^{n^2\,k}\) and therefore \(t\) consumes a word \(b^{N}\)
where \(N \geq 2^{n^2\,k}\) since each action of \(t\) consumes a \(b\).
\qed
\end{proof}

The reader can find in the Appendix a depiction of the accepting actree corresponding to \(P(2,1)\).

\begin{theorem} \label{thm:parikh_cfg}
For each \(n \geq 1\) and \(p> 2n+4\), there is a PDA with \(n\) states and \(p\) stack symbols for which every Parikh-equivalent CFG has \(\Omega(n^2(p - 2n - 4))\) variables.
\end{theorem}
\begin{proof}
Consider the family of PDAs \(P(n,k)\) with \(n \geq 1\) and \(k \geq 1\) described in Definition~\ref{def:P}.
Fix \(n\) and \(k\) and refer to the corresponding member of the family as \(P\).

First, Lemma \ref{lemma:unique_w} shows that \(L(P)\) consists of a single word \(b^N\) with \(N \geq 2^{n^2\,k}\). It follows that a language \(L\) is Parikh-equivalent to \(L(P)\) if{}f \(L\) is language-equivalent to \(L(P)\).

Let \(G\) be a CFG such that \(L(G) = L(P)\). 
The smallest CFG that generates exactly one word of length \(\ell\) has size \(\Omega(log(\ell))\)~\cite[Lemma~1]{Charikar2005}, where the size of a grammar is the sum of the length of all the rules. 
It follows that \(G\) is of size \(\Omega(log(2^{n^2k})) = \Omega({n^2k})\). 
As \(k = p - 2n - 4\), then \(G\) has size \(\Omega({n^2(p - 2n - 4)})\). 
We conclude that \(G\) has \(\Omega({n^2\,(p - 2n - 4)})\) variables.
\qed
\end{proof}

\begin{remark}
\label{rem:CFG}
According to the classical conversion algorithm, every CFG that is equivalent to \(P(n,k)\) needs at least \(n^2(k + 2n + 4) \in \mathcal{O}(n^2k + n^3)\) \hyphenation{va-ria-bles}{variables}. On the other hand, Theorem \ref{thm:parikh_cfg} shows that a lower bound for the number of variables is \(\Omega({n^2k})\).
We observe that, as long as \(n \leq Ck\) for some positive constant \(C\), the family \(P(n,k)\) shows that the conversion algorithm is optimal \footnote{Note that if \(n \leq Ck\) for some \(C>0\) then the \(n^3\) addend in \(\mathcal{O}(n^2k + n^3)\) becomes negligible compared to \(n^2k\), and the lower and upper bound coincide.} in the number of variables when assuming both language and Parikh equivalence.
Otherwise, the algorithm is not optimal as there exists a gap between the lower bound and the upper bound. For instance, if \(n = k^2\) then the upper bound is \mbox{\(\mathcal{O}(k^{5} + k^6) = \mathcal{O}(k^6)\)} while the lower bound is \(\Omega(k^5)\).
\end{remark}
\subsection{The Case of Unary Deterministic Pushdown Automata}\label{sec:specialcase}

We have seen that the classical translation from PDA to CFG is optimal in the number of grammar variables for the family of unary nondeterministic PDA \(P(n,k)\) when \(n\) is in linear relation with respect to \(k\) (see Remark \ref{rem:CFG}).
However, for unary \emph{deterministic} PDA (UDPDA for short) the situation is different.
Pighizzini~\cite{Pighizzini2009} shows that for every UDPDA with \(n\) states and \(p\) stack symbols, there exists an equivalent CFG with at most \(2np\) variables.
Although he gives a definition of such a grammar, we were not able to extract an algorithm from it.
On the other hand, Chistikov and Majumdar~\cite{Chistikov2014} give a polynomial time algorithm that transforms a UDPDA into an equivalent CFG going through the construction of a pair of straight-line programs.
The size of the resulting CFG is linear in that of the UDPDA.

We propose a new polynomial time algorithm that converts a UDPDA with \(n\) states and \(p\) stack symbols into an equivalent CFG with \(\mathcal{O}(np)\) variables.
Our algorithm is based on the observation that the conversion algorithm from PDAs to CFGs need not consider all the triples in \eqref{eq:triples}.
We discard unnecessary triples using the \emph{saturation procedure}~\cite{Bouajjani1997,Finkel1997} that computes the set of reachable IDs. 

For a given PDA \(P\) with \(q \in Q\) and \(X \in \Gamma\), define the set of \emph{reachable IDs} \(R_P(q,X)\) as follows:
\[R_P(q,X) = \{(q',\beta) \mid \exists (q,X) \vdash \cdots \vdash (q', \beta)\}\enspace .\]
\begin{lemma}
\label{lemma:udpda}
If \(P\) is a UDPDA then the set \(\{ I \in R_P(q,X) \mid \stack(I)=\varepsilon \}\) has at most one element for every state \(q\) and stack symbol \(X\).
\end{lemma}

\begin{proof}
Let \(P\) be a UDPDA with \(\Sigma = \{a\}\). Since \(P\) is \emph{deterministic} we have that \((i)\) for every \(q\in Q, X\in\Gamma\) and \(b\in\Sigma\cup\{\varepsilon\}\), \(\len{\delta(q,b,X)}\leq 1\) and, \((ii)\) for every \(q\in Q\) and \(X\in\Gamma\), if \(\delta(q,\varepsilon,X) \neq \emptyset\) then \(\delta(q,b,X) = \emptyset\) for every \(b\in \Sigma\).

The proof goes by contradiction. Assume that for some state \(q\) and stack symbol \(X\), there are two IDs \(I_1\) and \(I_2\) in \(R_P(q,X)\) such that 
\(\stack(I_1)=\stack(I_2)=\varepsilon\) and \(\state(I_1)\neq\state(I_2)\).

Necessarily, there exists three IDs \(J\), \(J_1\) and \(J_2\) with \(J_1 \neq J_2\) such that the following holds: 
\begin{align*}
(q,X) \vdash \cdots \vdash &J \vdash_{a} J_1 \vdash \cdots \vdash I_1\\
(q,X) \vdash \cdots \vdash &J \vdash_{b} J_2 \vdash \cdots \vdash I_2 \enspace .
\end{align*}
It is routine to check that if \(a=b\) then \(P\) is not deterministic, a contradiction. 
Next, we consider the case \(a\neq b\). 
When \(a\) and \(b\) are symbols, because \(P\) is a unary DPDA, then they are the same, a contradiction.
Else if either \(a\) or \(b\) is \(\varepsilon\) then \(P\) is not deterministic, a contradiction. 
We conclude from the previous that when \(\stack(I_1)=\stack(I_2)=\varepsilon\), then necessarily \(\state(I_1)=\state(I_2)\) and therefore 
that the set \(\{ I \in R_P(q,X) \mid \stack(I)=\varepsilon \}\) has at most one element.
\qed
\end{proof}

Intuitively, Lemma~\ref{lemma:udpda} shows that, when
fixing \(q\) and \(X\), there is at most one \(q'\) such that the triple
\([qXq']\) generates a string of terminals.  We use this fact to prove the following theorem.
\begin{theorem}
\label{thm:empty-udpda}
For every UDPDA with \(n\) states and \(p\) stack symbols, there is a polynomial time algorithm that computes an equivalent CFG with at most \(np\) variables.
\end{theorem}
\begin{proof}
The conversion algorithm translating a PDA \(P\) to a CFG \(G\) computes the set of grammar variables \( \{ [q X q'] \mid q, q'\in Q, X\in\Gamma\}\).
By Lemma~\ref{lemma:udpda}, for each \(q\) and \(X\) there is at most one variable \([qXq']\) in the previous set generating a string of terminals.
The consequence of the lemma is twofold:
\begin{inparaenum}[\upshape(\itshape i\upshape)]
\item For the triples it suffices to compute the subset \(T\) of the aforementioned generating variables. Clearly, \(\len{T}\leq np\).
\item Each action of \(P\) now yields a single rule in \(G\).
	This is because in \eqref{eq:rulesoftriples} there is at most one choice for
	\(r_1\) to \(r_d\), hence we avoid the exponential blowup of the runtime in the conversion algorithm.
\end{inparaenum}
To compute \(T\) given \(P\), we use the polynomial time saturation procedure~\cite{Bouajjani1997,Finkel1997} which given \( (q,X) \) computes a FSA for the set \(R_P(q,X)\). 
Then we compute from this set the unique state \(q'\) (if any) such that \( (q',\varepsilon) \in R_P(q,X)\), hence \(T\).
From the above we find that, given \(P\), we compute \(G\) in polynomial time.
\qed
\end{proof}

Up to this point, we have assumed the empty stack as the acceptance condition. 
For general PDA, assuming final states or empty stack as acceptance condition induces no loss of generality.
The situation is different for deterministic PDA where accepting by final states is more general than empty stack. 
For this reason, we contemplate the case where the UDPDA accepts by final states.
Theorem~\ref{thm:final-udpda} shows how our previous construction can be modified to accommodate the acceptance condition by final states.
\begin{theorem}
\label{thm:final-udpda}
For every UDPDA with \(n\) states and \(p\) stack symbols that accepts by final states, there is a polynomial time algorithm that computes an equivalent CFG with \(\mathcal{O}(np)\) variables.
\end{theorem}
\begin{proof}
Let \(P\) be a UDPDA with \(n\) states and \(p\) stack symbols that accepts by final states.
We first translate \(P=(Q,\Sigma,\Gamma,\delta,q_0,Z_0,F)\)\footnote{The set of final states is given by \(F\subseteq Q\).} into a (possibly nondeterministic) unary pushdown automaton \(P'=(Q',\Sigma,\Gamma',\delta',q'_0,Z'_0)\) with an empty stack acceptance condition.
In particular, \(Q'=Q\cup\{q'_0,\mathit{sink}\}\); \(\Gamma' = \Gamma\cup{Z'_0}\); and
\(\delta'\) is given by
\begin{align*}
\delta 
& \cup \{ (q'_0,Z'_0) \hookrightarrow_{\varepsilon} (q_0, Z_0\, Z'_0) \}\\
& \cup \{ (q,X) \hookrightarrow_{\varepsilon} (\mathit{sink}, X) \mid X \in \Gamma', q\in F\} \\
& \cup \{ (\mathit{sink},X) \hookrightarrow_{\varepsilon} (\mathit{sink},\varepsilon) \mid X\in\Gamma' \} \enspace .
\end{align*}
The new stack symbol \(Z'_0\) is to prevent \(P'\) from incorrectly accepting when \(P\) is in a nonfinal state with an empty stack.
The state \(\mathit{sink}\) is to empty the stack upon \(P\) entering a final state.
Observe that \(P'\) need not be deterministic.
Also, it is routine to check that \(L(P')=L(P)\) and \(P'\) is computable in time linear in the size of \(P\).
Now let us turn to \(R_{P'}(q,X)\).
For \(P'\) a weaker version of Lemma~\ref{lemma:udpda} holds:
the set \(H=\{I\in R_{P'}(q,X) \mid \stack(I)=\varepsilon\}\) has at most two elements for every state \(q\in Q'\) and stack symbols \(X\in\Gamma'\).
This is because if \(H\) contains two IDs then necessarily one of them has \(\mathit{sink}\) for state. 

Based on this result, we construct \(T\) as in Theorem~\ref{thm:empty-udpda}, but this time we have that \(\len{T}\) is \(\mathcal{O}(np)\). 

Now we turn to the set of production rules as defined in \eqref{eq:rulesoftriples} (see Section~\ref{sec:prelim}). 
We show that each action \((q,X) \hookrightarrow_b (q',\beta)\) of \(P'\) yields at most \(d\) production rules in \(G\) where \(d=\len{\beta}\).
For each state \(r_i\) in \eqref{eq:rulesoftriples} we have two choices, one of which is \(\mathit{sink}\). 
We also know that once a move sequence enters \(\mathit{sink}\) it cannot leave it.
Therefore, we have that if \(r_i = \mathit{sink}\) then \(r_{i+1}=\cdots=r_d=\mathit{sink}\).
Given an action, it thus yields \(d\) production rules one where \(r_1=\cdots=r_d=\mathit{sink}\), another where \(r_2=\cdots=r_d=\mathit{sink}\), \dots, etc.
Hence, we avoid the exponential blowup of the runtime in the conversion algorithm.

The remainder of the proof follows that of Theorem~\ref{thm:empty-udpda}.
\qed
\end{proof}

\section{Parikh-Equivalent Finite State Automata}\label{sec:FSA}

Parikh's theorem~\cite{Parikh66} shows that every context-free language is Parikh-equivalent to a regular language.
Using this result, we can compare PDAs against FSAs under Parikh equivalence.
We start by deriving some lower bound using the family \(P(n,k)\).
Because its alphabet is unary and it accepts a single long word, the comparison becomes straightforward.

\begin{theorem}
\label{thm:lowerbound}
For each \(n \geq 1\) and \(p > 2n + 4\), there is a PDA with \(n\) states and \(p\) stack symbols for which every Parikh-equivalent FSA  has at least \(2^{n^2(p - 2n - 4)} + 1\) states.
\end{theorem}
\begin{proof}
Consider the family of PDAs \(P(n,k)\) with \(n \geq 1\) and \(k \geq 1\) described in Definition~\ref{def:P}.
Fix \(n\) and \(k\) and refer to the corresponding member of the family as \(P\).
By Lemma~\ref{lemma:unique_w}, \(L(P) = \{b^N\}\) with \(N\geq 2 ^{n^2k}\). Then, the smallest FSA that is Parikh-equivalent to \(L(P)\) needs \(N+1\) states.
As \(k = p-2n-4\), we conclude that the smallest Parikh-equivalent FSA has at least \(2^{n^2(p - 2n - 4)} + 1\) states.
\qed
\end{proof}

Let us now turn to upper bounds. 
We give a 2-step procedure computing, given a PDA, a Parikh-equivalent FSA.
The steps are:
\begin{inparaenum}[\upshape(\itshape i\upshape)]
\item translate the PDA into a language-equivalent context-free grammar~\cite{AUTTH}; and
	\item translate the context-free grammar into a Parikh-equivalent finite state automaton~\cite{IPL}. 
\end{inparaenum}
Let us introduce the following definition.
A grammar is in \emph{2-1 normal form} (2-1-NF for sort) if each rule \( (X,\alpha) \in R\) is such that \(\alpha\) consists of at most one terminal and at most two variables. 
It is worth pointing that, when the grammar is in 2-1-NF, the resulting Parikh-equivalent FSA from step {\upshape(\itshape ii\upshape)} has \(\mathcal{O}(4^n)\) states where \(n\) is the number of grammar variables \cite{IPL}.
For the sake of simplicity, we will assume that grammars are in 2-1-NF which holds when PDAs are in \emph{reduced form}: every move is of the form \((q,X)\hookrightarrow_b (q',\beta)\) with \(\len{\beta}\leq 2\) and \(b \in \Sigma\cup\{\varepsilon\}\).

\begin{theorem}\label{thm:upperfsa}
Given a PDA in reduced form with \(n\geq 1\) states and \(p \geq 1\) stack symbols, there is a \mbox{Parikh-equivalent} FSA with \(\mathcal{O}(4^{n^2p})\) states.
\end{theorem}
\begin{proof} 
The algorithm to convert a PDA with \(n\geq1\) states and \(p\geq1\) stack symbols into a CFG that generates the same language~\cite{AUTTH} uses at most \(n^2p+1\) variables if \(n> 1\) (or \(p\) if \(n=1\)).
Given a CFG of \(n\) variables in 2-1-NF, one can construct a Parikh-equivalent FSA with \(\mathcal{O}(4^n)\) states~\cite{IPL}.

Given a PDA \(P\) with \(n\geq1\) states and \(p\geq1\) stack symbols the conversion algorithm returns a language-equivalent CFG \(G\). 
Note that if \(P\) is in reduced form, then the conversion algorithm returns a CFG in 2-1-NF.
Then, apply to \(G\) the known construction that builds a Parikh-equivalent FSA~\cite{IPL}.
The resulting FSA has \(\mathcal{O}(4^{n^2p})\) states.
\qed
\end{proof}

\begin{remark}
\label{rem:FSA}
Theorem \ref{thm:lowerbound} shows that a every FSA that is Parikh-equivalent to \(P(n,k)\) needs \(\Omega({2^{n^2k}})\) states. On the other hand, Theorem \ref{thm:upperfsa} shows that the number of states of every Parikh-equivalent FSA is \(O(4^{n^2(k + 2n + 4)})\). Thus, our construction is \emph{close to optimal}\footnote{As the blow up of our construction is \(O(4^{n^2(k + 2n + 4)})\) for a lower bound of \(2^{n^2k}\), we say that it is \emph{close to optimal} in the sense that \(2n^2(k + 2n + 4) \in \Theta(n^2k)\), which holds when \(n\) is in linear relation with respect to \(k\) (see Remark \ref{rem:CFG}).} when \(n\) is in linear relation with respect to \(k\).
\end{remark}

We conclude by discussing the reduced form assumption.
Its role is to simplify the exposition and, indeed, it is not needed to prove correctness of the 2-step procedure. 
The assumption can be relaxed and bounds can be inferred.  
They will contain an additional parameter related to the length of the longest sequence of symbols pushed on the stack.

\subsubsection{\ackname}
We thank Pedro Valero for pointing out the reference on smallest grammar problems \cite{Charikar2005}.
We also thank the anonymous referees for their insightful comments and suggestions.

\bibliographystyle{abbrv}

\begin{thebibliography}{10}

\bibitem{Bouajjani1997}
A.~Bouajjani, J.~Esparza, and O.~Maler.
\newblock Reachability analysis of pushdown automata: Application to
  model-checking.
\newblock In {\em {CONCUR}}, pages 135--150. Springer, 1997.

\bibitem{Charikar2005}
M.~Charikar, E.~Lehman, D.~Liu, R.~Panigrahy, M.~Prabhakaran, A.~Sahai, and
  A.~Shelat.
\newblock The smallest grammar problem.
\newblock {\em {IEEE} Transactions on Information Theory}, 51(7):2554--2576,
  2005.

\bibitem{Chistikov2014}
D.~Chistikov and R.~Majumdar.
\newblock Unary pushdown automata and straight-line programs.
\newblock In {\em {ICALP}}, pages 146--157. Springer, 2014.

\bibitem{IPL}
J.~Esparza, P.~Ganty, S.~Kiefer, and M.~Luttenberger.
\newblock Parikh's theorem: A simple and direct automaton construction.
\newblock {\em IPL}, pages 614--619, 2011.

\bibitem{Esparza2014}
J.~Esparza, M.~Luttenberger, and M.~Schlund.
\newblock A brief history of strahler numbers.
\newblock In {\em {LATA}}, pages 1--13. Springer, 2014.

\bibitem{Finkel1997}
A.~Finkel, B.~Willems, and P.~Wolper.
\newblock A direct symbolic approach to model checking pushdown systems
  (extended abstract).
\newblock {\em Electronic Notes in Theoretical Computer Science}, 9:27--37,
  1997.

\bibitem{GV2016}
P.~Ganty and D.~Valput.
\newblock Bounded-oscillation pushdown automata.
\newblock {\em EPTCS}, pages 178--197, 2016.
\newblock GandALF.

\bibitem{GOLD}
J.~Goldstine, J.~K. Price, and D.~Wotschke.
\newblock A pushdown automaton or a context-free grammar: Which is more
  economical?
\newblock {\em Theoretical Computer Science}, pages 33--40, 1982.

\bibitem{AUTTH}
J.~E. Hopcroft, R.~Motwani, and J.~D. Ullman.
\newblock {\em Introduction to Automata Theory, Languages, and Computation (3rd
  Edition)}.
\newblock Addison-Wesley Longman Publishing Co., Inc., Boston, MA, USA, 2006.

\bibitem{Pighizzini2009}
G.~Pighizzini.
\newblock Deterministic pushdown automata and unary languages.
\newblock {\em International Journal of Foundations of Computer Science},
  20(04):629--645, 2009.

\bibitem{Parikh66}
J.~P. Rohit.
\newblock On context-free languages.
\newblock {\em Journal of the ACM}, 13(4):570--581, 1966.

\end{thebibliography}
\newpage
\appendix 
\section{Appendix}

\subsection{Proof of Lemma~\ref{lemma:length}}
\begin{proof}
By induction on \(\len{t}\).

\noindent
\textbf{Base case.} Since \(\len{t}=1\) necessarily \(t = a\) and \(d(t) = 0\). Hence \(1 \geq 2^0\).

\noindent 
\textbf{Inductive case.}
Let \(t = {a}(a_1(\ldots),\ldots, a_r(\dots))\) with \(r\geq 1\). We study two cases.
Suppose there is a unique subtree \(t_x=a_x(\ldots)\) of \(t\) with \(x \in \{1, \ldots, r\}\) such that \( d(t_x) = d(t)\). 
As \(\len{t_x}<\len{t}\), the induction hypothesis shows that \(\len{t_x} \geq 2^{d(t_x)}=2^{d(t)}\), hence \(\len{t} \geq 2^{d(t)}\).

Next, let \(r\geq 2\) and suppose there are at least two subtrees \(t_x=a_x(\ldots)\) and \(t_y=a_y(\ldots)\) of \(t\) with \(x, y \in \{1, \ldots, r\}\) and \(x \neq y\) such that \(d(t_x) = d(t_y) = d(t)-1 \). 
As \(\len{t_x}<\len{t}\), the induction hypothesis shows that \(\len{t_x}\geq 2^{d(t_x)}\).
Applying the same reasoning to \(t_y\) we conclude from \(\len{t} \geq \len{t_x}+\len{t_y}\) that \(\len{t} \geq 2^{d(t_x)} + 2^{d(t_y)} = 2\cdot2^{d(t)-1} = 2^{d(t)}\).
\qed
\end{proof}

\subsection{Disassembly of Quasi-runs}
A quasi-run with more than one move can be \emph{disassembled} into its first move and subsequent quasi-runs.
To this end, we need to introduce a few auxiliary definitions.
Given a word \(w\in \Sigma^{*}\) and an integer \(i\), define \(w_{\mathit{sh}(i)} = (w)_{i+1}\cdots(w)_{i+\len{w}}\).
Intuitively, \(w\) is shifted \(i\) positions to the left if \(i\geq 0\) and to the right otherwise.
So given \(i\geq 0\), we will conveniently write \(w_{\ll_i}\) for \(w_{\mathit{sh}(i)}\) and \(w_{\gg_i}\) for \(w_{\mathit{sh}(-i)}\).
Moreover, set \(w_{\ll} = w_{\ll_1}\).
For example, \(a_{\ll_1} = a_{\gg_1} = \varepsilon\), \(abcde_{\ll_3} = de\), \(abcde_{\gg_3} = ab\), \(w = (w)_1 \cdots(w)_i\, w_{\ll_i}\) and \( w = w_{\gg_i}\, (w)_{\len{w}-i+1}\cdots (w)_{\len{w}}\) for \(i>0\).

Given an ID \(I\) and \(i>0\) define \(I_{\gg_i} = (\state(I),\tape(I),\stack(I)_{\gg_i})\) which, intuitively, removes from \(I\) the \(i\) bottom stack symbols.

\begin{lemma}[from~\cite{GV2016}]
\label{lemma:disassembly}
Let \(r = I_0 \vdash{\cdots}\vdash I_n\), be a quasi-run. Then we can \emph{disassemble} \(r\) into its first move \(I_0 \vdash I_1\) and \(d = \len{\stack(I_1)}\) quasi-runs \(r_1,\ldots,r_d\) each of which is such that
\[
r_i = (I_{p_{i-1}})_{\gg_{n_{i}}} \vdash{\cdots}\vdash (I_{p_i})_{\gg_{n_i}}\enspace .
\]
where \(p_0 \leq p_1 \leq \cdots \leq p_d\) are defined to be the least positions such that \(p_0=1\) and \(\stack(I_{p_i}) = \stack(I_{p_{i-1}})_{\ll}\) for all \(i\). 
Also \(n_i=\len{\stack(I_{p_i})}\) for all \(i\), that is \(r_i\) is
a quasi-run obtained by removing from the move sequence \(I_{p_{i-1}}\vdash{\cdots}\vdash I_{p_i}\) the \(n_i\) bottom stack symbols leaving the stack of \(I_{p_i}\) empty and that of \(I_{p_{i-1}}\) with one symbol only.
Necessarily, \(p_d = n\) and each quasi-run \(r_i\) starts with \((\stack(I_1))_i\) as its initial content.
\end{lemma}

\begin{example}
\label{ex:disassembly}
Recall the PDA \(P\) described in Example~\ref{ex:tree}. Consider the quasi-run:
\begin{align*}
r &= (q_0, X_1) \vdash (q_0,  X_0\,X_0) \vdash (q_1,  X_1\,\star\,X_0) \vdash (q_1,  X_0\,X_0\,\star\,X_0) \vdash (q_1, X_0\,\star\,X_0) \vdash\\
&(q_1,  \star\,X_0) \vdash (q_0, X_0) \vdash (q_1,  X_1\,\star) \vdash (q_1,  X_0\,X_0\,\star)\vdash (q_1, X_0\,\star)\vdash (q_1, \star) \vdash (q_0, \varepsilon)\enspace .
\end{align*}

 We can dissasemble \(r\) into its first move \(I_0 \vdash I_1 = (q_0, X_1) \vdash (q_0, X_0\,X_0)\) and \(d=2\) quasi-runs \(r_1, r_2\) such that
\begin{align*}
r_1 &= (I_{p_0})_{\gg_{n_1}} \vdash^{*} (I_{p_1})_{\gg_{n_1}} & p_0=1, p_1=6, n_1 = \len{\stack(I_6)} = 1\\
    &= (I_{1})_{\gg_1} \vdash^{*} (I_{6})_{\gg_1}\\
    &= (q_0, X_0) \vdash (q_1, X_1\,\star)\vdash (q_1,  X_0\,X_0\,\star)\vdash\\
    &\qquad (q_1, X_0\,\star) \vdash (q_1, \star)\vdash (q_0, \varepsilon)\\
r_2 &= (I_{p_1})_{\gg_{n_2}} \vdash^{*} (I_{p_2})_{\gg_{n_2}} & p_1=6, p_2=11, n_2 = \len{\stack(I_{11})} = 0\\
    &= (I_{6})_{\gg_0} \vdash (I_{11})_{\gg_0} \\
    &= (q_0, X_0) \vdash (q_1, X_1\,\star) \vdash (q_1, X_0\,X_0\,\star)\vdash\\
    &\qquad  (q_1, X_0\,\star) \vdash (q_1, \star) \vdash (q_0, \varepsilon)\enspace .
\end{align*}
Note that for each  quasi-run \(r_i\,(i=1,2)\), the stack of \((I_{p_i})_{\gg_{n_i}}\) is empty and that of \((I_{p_{i-1}})_{\gg_{n_i}}\) contains one symbol only. Also, \(p_d = p_2  = n = 11\) and each \(r_i\) starts with \((\stack(I_1))_i\) as its initial content.
\end{example}

\subsection{Assembly of Quasi-runs}

Now we show how to \emph{assemble} a quasi-run from a given action and a list of quasi-runs. We need the following notation: given \(I\) and \(w\in\Gamma^*\), define \(I \bullet w = (\state(I), \stack(I)\, w) \).

\begin{lemma}
\label{lemma:assembly}
Let \( a=(q,X)\hookrightarrow(q',\beta_1\ldots \beta_d)\) be an action and \(r_1, \ldots, r_d\) be \(d \geq 0\) quasi-runs with \(r_i = I^i_0 \vdash I^i_1 \vdash{\cdots}\vdash I^i_{n_i}\) for all \(i\), such that
\begin{compactitem}
  \item the first action of \(r_i\) pops \( \beta_i \) for every \(i\);
  \item the target state of last action of \(r_i\) (\(a\) when \(i=0\) ) is the source state of first action of \(r_{i+1}\) for all \( i \in \{1,\ldots,d-1\}\).
\end{compactitem}
Then there exists a quasi-run \(r\) given by
\begin{align}
  (q,X) \vdash (q',\beta_1\ldots\beta_d) &\vdash (I^1_1 \bullet \beta_2\ldots \beta_d)\vdash \cdots \vdash (I^1_{n_1} \bullet \beta_2\ldots \beta_d)\vdash\cdots \notag\\
  &\vdash(I^{\ell}_1 \bullet \beta_{\ell+1}\ldots \beta_d)\vdash \cdots \vdash(I^{\ell}_{n_\ell} \bullet \beta_{\ell+1}\ldots \beta_d)\vdash \cdots \notag\\
  &\vdash (I^d_1  \bullet \varepsilon)\vdash \cdots \vdash(I^d_{n_d} \bullet \varepsilon)\enspace .
\end{align}
\end{lemma}

\subsection{Proof of Theorem~\ref{thm:corespondence}}

\begin{proof} To prove the existence of a one-to-one correspondence we show that:
  \begin{compactenum}
    \item Each quasi-run must be paired with at least one actree, and viceversa.
    \item No quasi-run may be paired with more than one actree, and viceversa.
  \end{compactenum}
  1. First, given a quasi-run \(r\) of \(P\), define a tree \(t\) inductively on the length of \(r\). We prove at the same time that \eqref{eq:1} holds for \(t\) which we show is an actree.
\begin{align}
\label{eq:1}
&\text{\(\seq{t}\) and the sequence of actions of \(r\) coincide}.
\end{align}
For the base case, necessarily \(r = I_0\vdash I_1\) and we define \(t\) as the leaf labeled by the action \( I_0 \hookrightarrow I_1 \). 
Clearly, \(t\) satisfies (\ref{eq:1}), hence \(t\) is an actree (\(t\) trivially verifies Definition~\ref{def:tree}).

Now consider the case \(r = I_0 \vdash I_1 \vdash{\cdots}\vdash I_n\) where \(n>1\), we define \(t\) as follows.
Lemma~\ref{lemma:disassembly} shows \(r\) disassembles into its first action \(a\) and \(d = \len{\stack(I_1)} \geq 1\) quasi-runs \(r_1, \ldots, r_d\).
The action \(a\) labels the root of \(t\) which has \(d\) children \(t_1\) to \(t_d\).
The subtrees \(t_1\) to \(t_d\) are defined applying the induction hypothesis on the quasi-runs \(r_1\) to \(r_d\), respectively.
From the induction hypothesis, \(t_1\) to \(t_d\) are actrees and each sequence \(\seq{t_i}\) coincide with the sequence of actions of the quasi-run \(r_i\). 
Moreover, Lemma~\ref{lemma:disassembly} shows that the state of the last ID of \(r_i\) coincides with the state of the first ID of \(r_{i+1}\) for all \(i\in\{1,\ldots,d-1\}\).
Also, the state of the first ID of \(r_1\) coincides with the state of \(I_1\). 
We conclude from above that \(t\) satisfies \eqref{eq:1} and that
\(t\) is an actree since it verifies Definition~\ref{def:tree}.

Second, given an actree \(t\) of \(P\), we define a move sequence \(r\) inductively on the height of \(t\). 
We prove at the same time that \eqref{eq:1} holds for \(r\) which we show is a quasi-run.
For the base case, we assume \(h(t) = 0\). 
Then, the root of \(t\) is a leaf labeled by an action \(a = I_0 \hookrightarrow I_1\) and we define \(r = I_0 \vdash I_1\). 
Clearly, \(r\) satisfies \eqref{eq:1} and is a quasi-run.

Now, assume that \(t\) has \(d\) children \(t_1\) to \(t_d\), we define \(r\) as follows. 
By the induction hypothesis, each subtree \(t_i\) for all \(i \in \{1, \dots,d\}\) defines a quasi-run \(r_i\) verifying \eqref{eq:1}. 
The definition of actree shows that the root of \(t\) pushes \(\beta_1\) to \(\beta_d\) which are popped by its \(d\) children.
By induction hypothesis each \(r_i\) for all \(i \in \{1, \dots,d\}\) thus starts by popping \(\beta_i\). 
Next it follows from the induction hypothesis and the definition of actree that the target state of the action given by the last move of \(r_{i}\) coincides with the source state of the action given by the first move of \(r_{i+1}\) for all \(i\in\{1,\ldots,d-1\}\).
Moreover, the target state of \(a\) coincides with the source state of the action given by the first move of \(r_1\). 
Thus, applying Lemma~\ref{lemma:assembly} to the action given by the root of \(t\) and \(r_1,\ldots,r_d\) yields the quasi-run \(r\) that satisfies \eqref{eq:1} following our previous remarks.\\
2. First, we prove that no quasi-run may be paired with more than one actree. The proof goes by contradiction. Given a move sequence \( I_0\vdash \ldots \vdash I_n\), define its sequence of actions \(a_1\ldots a_n\) such that the move \(I_{i}\vdash I_{i+1}\) is given by the action \(a_{i+1}\), for all \(i\).
Note that two quasi-runs \(r = I_0\vdash\ldots \vdash I_n\) and \(r' = I'_0 \vdash \ldots \vdash I'_m\) are \emph{equal} if{}f their sequences of actions coincide.

Suppose that given the actrees \(t\) and \(t'\) with \(t \neq t'\), there exist two quasi-runs \(r\) and \(r'\) such that \(r\) is paired with \(t\) and \(r'\) is paired with \(t'\), under the relation we described in part 1. of this proof, and \(r = r'\).
Let \(\seq{t} = a_1,\ldots,a_n\) and \(\seq{t'} =  a'_1,\ldots,a'_m\).
Let \(p\in\{1,\dots, min(n,m)\}\) be the least position in both sequences such that \(a_p \neq a'_p\).
By \eqref{eq:1}, the sequences of actions of \(r\) and \(r'\) also differ at position \(p\) (at least).
Thus, \(r \neq r'\) (contradiction).

Second, we prove that no actree may be paired with more than one quasi-run. Again, we give a proof by contradiction.

Suppose that given the quasi-runs \(r\) and \(r'\) with \(r \neq r'\), there exist two actrees \(t\) and \(t'\) such that \(t\) is paired with \(r\) and \(t'\) is paired with \(r'\), under the relation we described in part 1. of the proof, and \(t = t'\).
We rely on the standard definition of equality between labeled trees.

Suppose \( a_1\ldots a_n\) is the sequence of actions of \(r\) and \(a'_1\ldots a'_m \) is the sequence of actions of \(r'\).
Let \(p\in\{1,\dots, min(n,m)\}\) be the least position such that \(a_p \neq a'_p\).
By \eqref{eq:1}, \(\seq{t}\) and \(\seq{t'}\) also differ at position \(p\) (at least).
Then, \(t \neq t'\) (contradiction).

\qed
\end{proof}

\subsection{Example: Accepting Actree of \texorpdfstring{\(P(2,1)\)}{P(2,1)} }
We give a graphical depiction of the accepting actree \(t\) of \(P(2,1)\).
Recall that \(P(2,1)\) corresponds to the member of the family \(P(n,k)\) that has \(2\) states \(q_0\) and \(q_1\), and \(9\) stack symbols \(S, X_0, X_1, s_0, s_1, r_0, r_1, \star\) and \(\$\).
Figure~\ref{fig:troot} represents \(t\) which has been split for layout reasons.
\vspace{-.75cm}
\begin{figure}[ht]
    \centering
    \begin{subfigure}[b]{1\textwidth}
		\caption{Top of the tree \(t\)}
    \label{fig:top}
		\begin{forest}
		[${(q_0, S)\hookrightarrow_b(q_0, X_1\,r_0)}$[${(q_0, X_1) \hookrightarrow_b (q_0, X_0\,r_0\,s_0\,X_0\,r_0)}$[{\framebox[1.5\width]{\(t_1\)}}][~${(q_0, r_0) \hookrightarrow_b (q_0, \varepsilon)}$][${(q_0, s_0) \hookrightarrow_b (q_0, \varepsilon)}$][{\framebox[1.5\width]{\(t_1\)}}][${(q_0, r_0) \hookrightarrow_b (q_0, \varepsilon)}$]][${(q_0, r_0) \hookrightarrow_b (q_0, \varepsilon)}$]]
	  	\end{forest}
    \end{subfigure}

    \begin{subfigure}[b]{1\textwidth}
		\caption{Subtree \(t_1\)}
    \label{fig:tr1}
		\begin{forest}
		[${(q_0, X_0) \hookrightarrow_b (q_0, X_1\,\star)}$[${(q_0, X_1) \hookrightarrow_b (q_0, X_0\,r_1\,s_0\,X_0\,r_1)}$[{\framebox[1.5\width]{\(t_2\)}}][~${(q_1, r_1) \hookrightarrow_b (q_1, \varepsilon)}$][${(q_1, s_0) \hookrightarrow_b (q_0, \varepsilon)}$][{\framebox[1.5\width]{\(t_2\)}}][${(q_1, r_1) \hookrightarrow_b (q_1, \varepsilon)}$]][${(q_1, \star) \hookrightarrow_b (q_0, \varepsilon)}$]]
	  \end{forest}
    \end{subfigure}

    \begin{subfigure}[b]{1\textwidth}
		\caption{Subtree \(t_2\)}
    \label{fig:tr2} 
		\begin{forest}
		[${(q_0, X_0)\hookrightarrow_b(q_1, X_1\,\$)}$[${(q_1, X_1) \hookrightarrow_b (q_1, X_0\,r_0\,s_1\,X_0\,r_0)}$[{\framebox[1.5\width]{\(t_3\)}}][~${(q_0, r_0) \hookrightarrow_b (q_0, \varepsilon)}$][${(q_0, s_1) \hookrightarrow_b (q_1, \varepsilon)}$][{\framebox[1.5\width]{\(t_3\)}}][${(q_0, r_0) \hookrightarrow_b (q_0, \varepsilon)}$]][${(q_0, \$) \hookrightarrow_b (q_1, \varepsilon)}$]]
	  \end{forest}
    \end{subfigure}

    \begin{subfigure}[b]{1\textwidth}
		\caption{Subtree \(t_3\)}
		\label{fig:tr3} 
		\begin{forest}
			[${(q_1, X_0) \hookrightarrow_b (q_1, X_1\star)}$[${(q_1, X_1) \hookrightarrow_b (q_1, X_0\,r_1\,s_1\,X_0\,r_1)}$[${(q_1, X_0) {\hookrightarrow_b}(q_1, \varepsilon)}$][${(q_1, r_1) {\hookrightarrow_b}(q_1, \varepsilon)}$][${(q_1, s_1) {\hookrightarrow_b}(q_1, \varepsilon)}$][${(q_1, X_0) {\hookrightarrow_b}(q_1, \varepsilon)}$][${(q_1, r_1) {\hookrightarrow_b}(q_1, \varepsilon)}$]][${(q_1, \star) {\hookrightarrow_b}(q_0, \varepsilon)}$]]
	  	\end{forest}
    \end{subfigure}
    \caption{Accepting actree \(t\) of \(P(2,1)\) split into 4 subtrees.}\label{fig:troot}
\end{figure}
\end{document}